\documentclass[11pt]{amsart}
\usepackage{amscd,amssymb,amsxtra}
\usepackage[mathscr]{eucal}
\usepackage{mathabx}
\usepackage{comment}
\usepackage{color}
\usepackage{enumitem}
\makeatletter
\let\@secnumfont\bfseries
\def\section{\@startsection{section}{1}%
  \z@{4\linespacing\@plus\linespacing}{\linespacing}%
  {\bfseries\centering}}
\def\introsection{\@startsection{section}{1}%
  \z@{3\linespacing\@plus\linespacing}{\linespacing}%
  {\bfseries\centering}}
\def\subsection{\@startsection{subsection}{2}%
   \z@{1.25\linespacing\@plus.7\linespacing}{.5\linespacing}%
   {\normalfont\bfseries}}
\def\subsectionsinline{\def\subsection{\@startsection{subsection}{2}%
  \z@{1\linespacing\@plus.7\linespacing}{-.5em}%
  {\normalfont\bfseries}}}

\makeatother

\theoremstyle{definition}
\newtheorem{definition}[equation]{Definition}

\newtheorem*{definition*}{Definition}
\newtheorem*{example*}{Example}
\newtheorem*{problem*}{Problem}
\newtheorem*{exercise*}{Exercise}
\newtheorem*{question*}{\color{blue}Question}
\newtheorem*{construction*}{Construction}

\theoremstyle{remark}

\newtheorem{remark}[equation]{Remark}

\newtheorem*{note*}{Note}
\newtheorem*{notation*}{Notation}
\newtheorem*{remark*}{Remark}
\newtheorem*{data*}{Data}

\theoremstyle{plain}
\newtheorem{theorem}[equation]{Theorem}

\newtheorem{lemma}[equation]{Lemma}
\newtheorem{proposition}[equation]{Proposition}

\newtheorem*{theorem*}{Theorem}
\newtheorem*{corollary*}{Corollary}
\newtheorem*{lemma*}{Lemma}
\newtheorem*{proposition*}{Proposition}
\newtheorem*{conjecture*}{Conjecture}
\newtheorem*{claim*}{Claim}
\newtheorem*{proposal*}{Proposal}
\newtheorem*{conclusion*}{Conclusion}
\newtheorem*{hypothesis*}{Hypothesis}
\newtheorem*{assumption*}{Assumption}

\newenvironment{proof*}[1][\proofname]{
  \begin{proof}[#1]}{  
\end{proof}}

\numberwithin{equation}{section}

\definecolor{refkey}{rgb}{0,.6,.4}

\renewcommand{\:}{\colon}

\newcommand{\Ahat}{{\hat A}}

\newcommand{\CC}{{\mathbb C}}
\newcommand{\CP}{{\mathbb C\mathbb P}}

\DeclareMathOperator{\Hom}{Hom}

\DeclareMathOperator{\pt}{pt}

\newcommand{\RR}{{\mathbb R}}
\newcommand{\TT}{\mathbb T}
\DeclareMathOperator{\Spin}{Spin}

\newcommand{\ZZ}{{\mathbb Z}}

\newcommand{\chiup}{\raise.5ex\hbox{$\chi$}}

\newcommand{\inv}{^{-1}}
\DeclareRobustCommand{\mstrut}{^{\vphantom{1*\prime y\vee M}}}

\newcommand{\temsquare}{\raise3.5pt\hbox{\boxed{ }}}

\newcommand{\zmod}[1]{\ZZ/#1\ZZ}

\newcommand{\zt}{\zmod2}

\usepackage[all,2cell]{xy}
\usepackage[colorlinks,citecolor=refkey]{hyperref}

\definecolor{refkey}{rgb}{0,.8,.2}\definecolor{labelkey}{rgb}{1,0,0}

\DeclareMathOperator{\Ext}{Ext}

\DeclareMathOperator{\spinc}{spin^{c}}
\newcommand{\Co}{\CP^1}

\newcommand{\IZ}{I\ZZ}
\newcommand{\MSpin}{M\!\Spin}
\newcommand{\Oo}{\mathcal{O}(1)}
\newcommand{\RZ}{\RR/\ZZ}
\newcommand{\Spcn}{\Spin^c_n}
\newcommand{\cE}{\widecheck{E}}
\newcommand{\ca}{\check\alpha }
\newcommand{\cc}{\check\chi }
\newcommand{\CM}{\mathcal{M}}
\newcommand{\CN}{\mathcal{N}}

\newcommand{\sB}{\mathscr{B}}
\newcommand{\sX}{\mathscr{X}}
\newcommand{\sY}{\mathscr{Y}}

\begin{document}

\abovedisplayskip18pt plus4.5pt minus9pt
\belowdisplayskip \abovedisplayskip
\abovedisplayshortskip0pt plus4.5pt
\belowdisplayshortskip10.5pt plus4.5pt minus6pt
\baselineskip=15 truept
\marginparwidth=55pt

\makeatletter
\renewcommand{\tocsection}[3]{%
  \indentlabel{\@ifempty{#2}{\hskip1.5em}{\ignorespaces#1 #2.\;\;}}#3}
\renewcommand{\tocsubsection}[3]{%
  \indentlabel{\@ifempty{#2}{\hskip 2.5em}{\hskip 2.5em\ignorespaces#1%
    #2.\;\;}}#3}
\makeatother

\setcounter{tocdepth}{2}




 \title[Topological sectors in the $\CP^1$ $\sigma $-model]{The sum over topological sectors \\ and $\theta$ in the 2+1-dimensional $\CP^1$ $\sigma$-model} 
 \author[D. S. Freed]{Daniel S.~Freed}
 \address{Department of Mathematics \\ University of Texas \\ Austin, TX
78712}
 \email{dafr@math.utexas.edu}
 \author[Z. Komargodski]{Zohar Komargodski}
 \address{Department of Particle Physics and Astrophysics, Weizmann Institute
of Science, ISRAEL\newline\hphantom{Mi}Simons Center for Geometry and Physics, Stony
Brook University,
Stony Brook, NY}
 \email{zohar.komargodski@weizmann.ac.il}
 \author[N. Seiberg]{Nathan Seiberg}
 \address{School of Natural Sciences, Institute for Advanced Study, Princeton, NJ 08540}
 \email{seiberg@ias.edu}
 \date{July 17, 2017}
 \begin{abstract}
 We discuss the three spacetime dimensional $\CP^N$ model and specialize to
the $\CP^1$ model.  Because of the Hopf map $\pi_3(\CP^1)=\ZZ$ one might try
to couple the model to a periodic $\theta$~ parameter. However, we argue that
only the values $\theta=0$ and $\theta=\pi$ are consistent. For these values
the Skyrmions in the model are bosons and fermions respectively, rather than
being anyons. We also extend the model by coupling it to a topological quantum
field theory, such that the Skyrmions are anyons.  We use techniques from
geometry and topology to construct the $\theta =\pi $ theory on arbitrary
3-manifolds, and use recent results about invertible field theories to prove
that no other values of~$\theta $ satisfy the necessary locality.
 \end{abstract}
\maketitle




   \section{Introduction}\label{sec:3}

The functional integral definition of quantum field theory involves
integrating over all possible configurations with a certain weight.  It is
often the case that the configuration space in the Euclidean functional
integral breaks into topologically distinct sectors labeled by $\nu$.  (These
sectors and their characterization can depend on the Euclidean spacetime the
theory is placed on.)  Then, defining $Z_\nu$ as the sum over the
configurations in the sector $\nu$, the total functional integral is given by
a linear combination of $Z_\nu$
 \begin{equation}\label{totalZ}
  Z=\sum_\nu a_\nu Z_\nu ~.
 \end{equation}
The possible values of the coefficients $a_\nu$ are constrained by various
consistency conditions like locality and unitarity.  Different consistent
choices of the $a_\nu$ correspond to distinct quantum field theories.  An
interesting problem is to find all possible consistent values of these
coefficients, thus finding all possible theories constructed out of the
building blocks $Z_\nu$.

A well known example is the quantum mechanical system of a single degree of
freedom on a circle.  Here, with Euclidean compact time the configuration
space is the space of maps $S^1\to S^1$ and $\nu$ is the winding number.  In
this case the coefficients $a_\nu$ are constrained to be determined by a
single periodic parameter $\theta$ as
 \begin{equation}\label{anut}
  a_\nu =e^{i\nu\theta}~.
 \end{equation}
Another example is the $4d$ pure $SU(N)$ gauge theory, where $\nu$ is the
instanton number and again we have \eqref{anut}. In these two cases we can
express $\nu$ as an integral of a local gauge invariant density and we can
interpret \eqref{anut} as arising from a term in the fundamental
Lagrangian. In many situations $\nu$ cannot be written as an integral over a
local density, but still an expression like \eqref{anut} exists.  A typical
example is the $1+1$-dimensional $SO(3)$ gauge theory, where $\nu$ is defined
modulo $2$ as the second Stiefel-Whitney class of a principal $SO(3)$-bundle,
and correspondingly the allowed values of $\theta$ in \eqref{anut} are $0$
and $ \pi$.

Locality and unitarity do not require $ a_\nu $ to be {the
exponential of} the integral of a local density, but rather they must be the
partition functions of an invertible field theory~\cite{FM}.  In physics
terms, $\log a_\nu$ can be thought of as an action of a classical field
theory, which is local, but not necessarily an integral of a local
density. Recent progress in understanding the structure of invertible
theories can be brought to bear on the problem of combining~$Z_\nu $ into a
well-defined theory.

One of the goals of this paper is to clarify this sum over sectors in the
$2+1$ dimensional nonlinear $\CP^1$ $\sigma$-model.  Placing the theory on
$S^3$ and using the Hopf invariant, which is associated with
$\pi_3(\CP^1)=\ZZ$, the label $\nu$ in \eqref{totalZ} runs over the integers.
It labels an instanton number.  Then one might think that \eqref{anut} is a
consistent prescription for how to sum over these sectors and the theory is
labeled by a continuous periodic parameter $\theta$. Explicitly, let $\vec
n^2=1$ be a coordinate on $\CP^1\simeq S^2$. Define ${\rm Hopf}(\vec n)$ to
be a density such that $\int d^3x\ {\rm Hopf}(\vec n)\in \ZZ$ is the Hopf
invariant. Then, we can modify the standard Euclidean Lagrangian for $\vec n$
by adding a theta term (see e.g.\
\cite{WilczekCY,PolyakovMD} and many followup papers where this term was
discussed) as follows
 \begin{equation}\label{LwithHopf}
  {\mathcal L} = {f\over 2} (\partial \vec n)^2+i\theta {\rm Hopf}(\vec n)~,
 \end{equation}
with a dimensionful parameter $f$. In this presentation it would seem that
any $\theta$ is allowed and only $\theta\mod 2\pi$ matters. A hint that something might be wrong with this $\theta$ term
comes from the fact that ${\rm Hopf}(\vec n)$ does not have a local
expression in terms of $\vec n$. Furthermore, it is unclear how to define
this theta term on other three-manifolds.  Indeed, it has been known
that $\theta=0,\pi$ naturally arise in simple situations but not the other
values of $\theta$. See e.g.~\cite{AbanovQZ} and references therein.

We will {prove that, in fact, only $\theta=0$ and $\pi$ are
consistent.\footnote{The authors of \cite{DzyaloshinskyPQ} noted that certain
microscopic 2+1 dimensional models of spins lead only to the values $\theta=
0$ and $\theta=\pi$. The same is true in 1+1 dimensions for such microscopic
models. But unlike our claimed result in 2+1 dimensions, in 1+1 dimensions
the $\CP^1$ model is well defined with arbitrary $\theta$ and not just at
$\theta=0,\pi$.} Furthermore, we will explicitly construct the corresponding
mod 2 invariant on arbitrary spin three-manifolds.  We will also present
variants of the $\CP^1$ model, where the low-energy $\CP^1$ Goldstone bosons
are coupled to a nontrivial TQFT leading to additional long range
interactions such that $\theta$ behaves as if it has other values. These
other values of $\theta$ are now allowed because we have modified the theory
in the deep infrared.  In condensed matter language, we could, for example,
think about that as coupling the $\CP^1$ theory to a fractional quantum hall
state.\footnote{We thank P. Wiegmann for many useful discussions.}

In section 2 we will discuss a microscopic theory that flows at long
distances to the $\CP^N$ model with a Wess-Zumino term.  Here we will see that
the $\CP^1$ model is special and this microscopic construction flows only to
$\theta=0$ or $\theta=\pi$.  In section 3 we will study the local operators
in the theory and we will argue that the $\CP^1$ model makes sense only for
these two values of $\theta$ and not for generic values.  In section 4 we
will present modifications of the $\CP^1$ model, which behave as if they have
other values of $\theta$.  This is consistent with the arguments in section
3, because the low energy theory is not simply the $\CP^1$ model, but it is
coupled to a TQFT.  In section 5 we bring to bear results about
invertible field theories.  These arguments are based on extended locality
and factorization and use an analysis of the theory on general Wick-rotated
spacetimes.  We freely employ techniques from global analysis and homotopy
theory to construct the allowed terms at $\theta =0$ and $\theta=\pi$ and
prove that these are the only allowed values of $\theta$.

   \section{The $2+1$ dimensional $\CP^N$ model from a linear model}\label{sec:4}

We find it instructive to view the $3d$ $\CP^1$ model as an atypical special
case of the $3d$ $\CP^N$ model.  We study\footnote{Since all our
three-manifolds are orientable, they admit a spin structure.  When we say
that they are not spin we mean that we do not pick a spin structure, and if
we do, the answers are independent of that choice.  When we later view the
three-manifold as a boundary of a four-manifold, the latter is only assumed
to be spin$^c$, not necessarily spin.} it on a closed (that is, compact and
without a boundary) spin$^c$ (not necessarily spin) manifold $\CM_3$. In this
section we will embed this model in a particular microscopic renormalizable
field theory, which flows at long distances to the nonlinear model with
particular Wess-Zumino terms.

We start with $N+1$ scalar fields coupled to a $U(1)$ gauge field $b$. The
Lagrangian is given by
  \begin{equation}\label{Linearmodel} {1\over
  2e^2}(db)^2+ \sum_{I=1}^{N+1}|D_b\phi^I|^2+\mu^2 \sum_{I=1}^{N+1}|\phi^I|^2+
  (\sum_{I=1}^{N+1}|\phi^I|^2)^2
  ~.\end{equation}
If the scalar fields condense (e.g.\ when $\mu^2<0$),
then the gauge field $b$ is Higgsed and the global $SU(N+1)$ symmetry is spontaneously 
broken to $S[U(N)\times U(1)]$. Therefore we obtain in the deep infrared
Goldstone bosons parameterising the coset ${SU(N+1)\over S[U(N)\times
U(1)]}=\CP^{N}$.

We could modify the model~\eqref{Linearmodel} and add to the action the Chern-Simons
terms\footnote{Here and in similar expressions below we are imprecise since
the connection forms are not global forms on $\CM_3$.}
 \begin{equation}\label{CSt}
 \begin{aligned}
 &\int_{\CM_3}\left({K\over 4\pi} bdb +{1\over 2\pi} db B\right) \qquad {\rm
for} \qquad K\in 2\ZZ \\
&\int_{\CM_3}\left({K\over 4\pi} bdb +{1\over 2\pi} db A\right) \qquad {\rm for} \qquad K\in 2\ZZ +1~,
 \end{aligned}
 \end{equation}
where $A$ is a classical background spin$^c$ connection and $B$ is a
classical background $U(1)$ gauge field (see \cite{SeibergRSG,SeibergGMD}
for more details).  For even $K$ each term in \eqref{CSt} is separately well
defined (up to an additive $2\pi i \ZZ$, which does not affect the exponential
of the action).  The same is true for odd $K$ on a spin manifold.  But on a
general spin$^c$ manifold with odd $K$ only the sum of the two terms in
\eqref{CSt} is well defined mod $2\pi i \ZZ$.

The terms \eqref{CSt} are made more precise by writing\footnote{One can
prove that $\CM_4$ and extensions of the gauge fields and spin$^c$ structure exist.} them as $4d$ integrals
 \begin{equation}\label{CSfdf}
 \begin{aligned}
 &\int_{\CM_4} \left({K\over 4\pi} dbdb +{1\over 2\pi} db dB \right)\qquad
{\rm for} \qquad K\in 2\ZZ \\
&\int_{\CM_4}\left({K\over 4\pi} dbdb +{1\over 2\pi} db dA\right) \qquad {\rm
for} \qquad K\in 2\ZZ +1 ~,
 \end{aligned}
 \end{equation}
where the original spacetime $\CM_3$ is the boundary of $\CM_4$.

The first term in \eqref{CSt} (or, equivalently, \eqref{CSfdf}) is a
Chern-Simons term.  Due to it, the monopole operator acquires spin $K/2$.
The second term means that the monopole operator of the theory carries charge
one under a global magnetic $U(1)$ symmetry.\footnote{A monopole operator is
defined by removing a point from our spacetime and specifying boundary
conditions on the $S^2$ around it such that $\int_{S^2} db = 2\pi$.  (There
are many such distinct operators.)  There are several ways to see that such
an operator carries $SU(2)$ spin $j\ge K/2$ with $j-K/2=0\mod 1$ (see
~\cite{BKW}, for example).}  Because of the Chern-Simons terms \eqref{CSt},
this monopole operator carries charge $K$. We can make the monopole gauge
invariant by multiplying it by $K$ charged scalars.  For simplicity, let all
of them be at the same point on the $S^2$.  This configuration is not
invariant under the $SU(2)$ isometry of the sphere. In order to have an
$SU(2)$ covariant description we replace the state with fixed position of the
scalars with another wave function -- we introduce collective coordinates for
the action of the symmetry.  Intuitively, they move the location of the
scalars and hence they take values in $S^2$.  (If the scalars are at
different positions, the collective coordinates are on the symmetric product
$Sym_K(S^2)$.) So the effective theory of the collective coordinates is a
quantum mechanical system with an $S^2$ target space.  The Chern-Simons
coupling \eqref{CSt} becomes a standard Wess-Zumino term in this quantum
mechanical system.  This is the familiar problem of a charge $K$ particle in
the background of a magnetic monopole. The result is that the system has spin
$K/2$ (and possible higher order excitations with higher spins, which are
$K/2 \;+$ integer).  } $A$ or $B$ are classical background fields for that
symmetry.  These monopole operators satisfy the spin/charge relation
\cite{SeibergRSG}.

  \subsection{$N>1$}\label{subsec:4.1}

For $N>1$ the analysis of the linear model~\eqref{Linearmodel} with the
Chern-Simons terms~\eqref{CSt} is completely standard.  With negative mass
squared for the scalars $\phi^I$ they obtain an expectation value.  Then
we integrate out $b$ by using its equations of motion to find at low energies
a nonlinear model on $\CP^N$.  The $b$ equation of motion sets $db=\omega +
\dots$, where the ellipses represent higher order terms in the inverse radius
of the target space ($\sqrt f$ in \eqref{LwithHopf}) that we will ignore and
$\omega$ is the pullback of the K\"ahler form of the $\CP^N$ target space.
We normalize it such that its periods are integer multiples of $2\pi$
 \begin{equation}\label{Mspace}
  \int_{\CM_2} \omega = 2 \pi \ZZ ~.
 \end{equation}
Then, \eqref{CSfdf} becomes
 \begin{equation}\label{CSfd}
 \begin{aligned}
 &\int_{\CM_4} \left({K\over 4\pi} \omega\omega  +{1\over 2\pi} \omega dB
\right)\qquad {\rm for} \qquad K\in 2\ZZ \\
&\int_{\CM_4}\left({K\over 4\pi} \omega\omega +{1\over 2\pi} \omega dA\right)
\qquad {\rm for} \qquad K\in 2\ZZ +1 ~.
 \end{aligned}
 \end{equation}

As we remarked after \eqref{CSt}, on a spin manifold each term in
\eqref{CSfd} is separately meaningful as a $3d$ term.  The same is true on a
general spin$^c$ manifold with even $K$.  But for odd $K$ on a general
spin$^c$ manifold only the sum of the two terms in \eqref{CSfd} is
meaningful.  This means that the first term, is not associated with $H^4$,
but with a more subtle cohomology~\cite{F2}; see~\S\ref{subsec:2.2}.

The first term in \eqref{CSfd} is a Wess-Zumino term of the nonlinear model.
The second term has two complementary interpretations.  The first
interpretation involves the solitons of the model, which are known as
Skyrmions.  Viewing $\CM_2$ in \eqref{Mspace} as our space, these are
configurations with nonzero $\int_{\CM_2} \omega$.  The second term in
\eqref{CSfd} means that they carry charge $\int_{\CM_2}\omega$ under the
global $U(1)$ symmetry that $A$ or $B$ couple to.  This is completely
analogous to the situation in the $4d$ chiral Lagrangian
\cite{WittenTW,WittenTX}, where the Skyrmions carry baryon number.  The
second interpretation is related to our discussion above of monopole
operators of the microscopic theory. Similarly, we can discuss Skyrmion
operators in the macroscopic $\CP^N$ theory.  They are defined by removing a
point from our spacetime $\CM_3$ and specifying boundary conditions on the
small $S^2$ around the point such that with $S^2=\CM_2$ in \eqref{Mspace} we
have $\int_{S^2}\omega=2\pi$.  These operators have spin $K/2$ and the second
term in \eqref{CSfd} means that they are charged under the global $U(1)$
symmetry.\footnote{The computation of their spin is very similar to the
computation of the spin of the monopole operators in the previous
footnote. This is not surprising because they are the macroscopic descendants
of the microscopic monopole operators.  More explicitly, the Skyrmion
configuration breaks the $SU(2)$ isometry of the $S^2$ and the $SU(2)$ global
symmetry of the target space.  We introduce collective coordinates for their
quantization.  As above, the Wess-Zumino term in the $3d$ problem becomes a
Wess-Zumino term in the $1d$ problem of the collective coordinates making the
spin of the state $K/2$ or larger.}  Again, we see here the spin/charge
relation \cite{SeibergRSG}.

  \subsection{$N=1$}\label{subsec:4.2}

Just as in the analogous $4d$ problem where the special case with two flavors
is slightly different, the same is true in our case for $N=1$.  In these two
situations there is no standard Wess-Zumino term, though there is a
nonstandard one, defined for all~$N$, which specializes for~$N=1$ to the
mod~2 invariant at~$\theta =\pi $; see~\S\ref{subsec:2.2}.

Starting with the same microscopic linear model with a Chern-Simons term
\eqref{CSfdf} we can follow the steps above to integrate out $b$.  We again
find $db=\omega + \dots$.  But here we cannot write \eqref{CSfd}.  There are two
related reasons for that.  First, the first term, involving $\omega\omega$,
clearly vanishes -- there is no four-form on $\CP^1$.  Second, in fact, we
cannot always extend the fields on $\CM_3$ to $\CM_4$.  Specifically, if
$\CM_3=S^3$ the $\CP^1$ configurations are labeled by an integer $\CN$
associated with the nontrivial Hopf invariant $\pi_3(\CP^1)=\ZZ$ and
configurations with odd $\CN$ cannot be extended to a $4d$ bulk (see
Remark~\ref{thm:6}).

This is analogous to the similar situation in $4d$, which is associated with
$\pi_4(S^3)=\ZZ_2$.  There the Wess-Zumino term is replaced by another
term\footnote{A uniform picture using E-cohomology was presented for the $4d$
problem in \cite{F2}, and will be discussed for the $3d$ problem in section
5.} representing this $\ZZ_2$.  It can be viewed as a discrete $\theta$
parameter term in the sigma model \cite{WittenTW,WittenTX}.  We can try to
imitate it in our problem and to use $\pi_3(\CP^1)=\ZZ$, which can be
expressed as
 \begin{equation}\label{CNb}
  \CN={1\over 4\pi^2} \int_{\CM_3}  b_0 d b_0~,
 \end{equation}
where $d b_0 =\omega$.  Then we can attempt to add to the action the
$\theta$-term
 \begin{equation}\label{thetaterm}
  \theta\CN= {\theta\over 4\pi^2} \int_{\CM_3}  b_0 d b_0~.
 \end{equation}
Of course, the subtlety in this expression is in the fact that $b_0$ is a
nonlocal expression in terms of the nonlinear model variables.  This fact is
at the root of our conclusion below that the theory is consistent only for
$\theta=0$ and $\theta=\pi$.  Such expressions were studied in
\cite{WilczekCY,PolyakovMD} and many followup papers.

Repeating the analysis for $N>1$ we see that the microscopic Chern-Simons
couplings \eqref{CSt} lead to $\theta = \pi K$; i.e.\ we find the $\CP^1$
theory with $\theta=0$ or $\pi$.  As for higher $N$, we can write
 \begin{equation}\label{CSfdb}
 \begin{aligned}
&{1\over 2\pi} \int_{\CM_3} \omega B \ \ \qquad\qquad {\rm for}\qquad K\in
2\ZZ\\
&\pi\CN +{1\over 2\pi} \int_{\CM_3} \omega A \qquad {\rm for}\qquad K\in
2\ZZ+1 ~.
 \end{aligned}
 \end{equation}
Therefore, for $\theta=0$ the Skyrmions are bosons and for $\theta=\pi$ they
are fermions \cite{WilczekCY,PolyakovMD}.

Such constructions have  recently appeared  in~\cite{KS}.

   \section{Skyrmion operators and conditions on $\theta$}\label{sec:5}

The previous discussion raises the following question.  Starting with a
microscopic theory we derived the $\CP^1$ model with $\theta=0$ or
$\theta=\pi$ \eqref{CSfdb}.  Are there microscopic models that lead to other
values of $\theta$ as in \eqref{thetaterm}?  In order to address this
question we should study the $\CP^1$ model without relying on the details of
its UV completion.

We are now going to argue that the $\CP^1$ model with generic $\theta$ is
not a consistent quantum field theory, thus explaining why we cannot derive
it from a microscopic model.

Consider the theory on $\CM_3=S^2\times \RR$ and view $S^2$ as space and
$\RR$ as time. The theory has a topologically conserved current $\star
\omega$, hence, the Hilbert space is decomposed into sectors with fixed
soliton number $\int_{S^2} \omega$.  Let us turn on \eqref{thetaterm} with a
generic value of $\theta$.  It was argued in \cite{WilczekCY} that in this case
the single Skyrmion has spin $\theta/2\pi$.  This is a valid
answer\footnote{Particles in 2+1 dimensions are in representations of the
universal cover of the little group $SO(2)$, i.e.\ they can have any real spin.}
for a particle on a spatial $\RR^2$, but it is not sensible for the system on
$S^2$.  The states in this case must be in representations of $SU(2)$ (i.e.\
the universal cover of the Euclidean isometry group, $SO(3)$).  This is the
case only for $\theta=0$ or $\pi$.

An equivalent way to state it in $\RR^3$ is the following.  Particle states
must be in representations of a multiple cover of the rotation group
$SO(2)$; i.e.\ they can be in any representation of $\RR$.  As such,
they can have arbitrary real spin.  Local operators, on the other hand, must
be in representations of the Lorentz group, which in Euclidean space is
$SU(2)$.  Now consider a Skyrmion operator.  It is defined by removing a
point from $\RR^3$ and specifying boundary conditions that the integral over
a small $S^2$ around it is $\int_{S^2}\omega =1$.  In the presence of a
$\theta$-term this operator has spin $\theta/2\pi$ and therefore $\theta$
should be $0$ or $\pi$.  The relation between this point about the local
operators and the previous argument based on quantization on a spatial $S^2$
is standard and is clear using radial quantization.

These two equivalent perspectives can be stated also in the following way.
Above we mentioned the effective quantum mechanics in a monopole or Skyrmion
sector.  When we quantize the system on a spatial $S^2$ in the sector with
$\int_{S^2} \omega=\int_{S^2} db_0=2\pi$ the effective quantum mechanical
problem includes a collective coordinate on $S^2$.  The term \eqref{thetaterm}
leads to a Wess-Zumino term in that quantum mechanical problem.  It is
$SU(2)$ invariant only for $\theta=0$ or $\theta=\pi$. In more detail, the
classical theory is well defined and is $SU(2)$ invariant for all $\theta$.
In the quantum theory we have two options for generic $\theta$.  We can have
a well defined, but not $SU(2)$ invariant expression (pick a point on the
$S^2$, connect it to the position of the particle and the Lagrangian is the
area swept by that line during the time $(t,t+dt)$) or we can have an $SU(2)$
invariant expression by extending the map to a higher dimension, but then the
answer depends on the choice of the extension.

We see that even though the partition function of the theory on $S^3$ allows
arbitrary $\theta$ the entire quantum field theory is consistent only for
$\theta=0$ or $\theta=\pi$.  In section 5 we will find the same conclusion by
imposing consistency of the theory on more complicated spacetimes~$\CM_3$.
Here we argued it using $\CM_3=\RR^3-\{ {\rm point}\}$, or equivalently by
studying the local operators in the theory.

This is reminiscent of the analysis of \cite{AharonyHDA}, where subtle
choices in the coefficients $a_\nu$ in \eqref{totalZ} and corresponding
subtle topological terms in the action were identified by studying a $4d$
theory on $\RR^4$ minus some lines; i.e.\ by studying the consistency of line
operators in the theory.

   \section{Changing the quantization of the coefficients}\label{sec:6}

Next, following \cite{SeibergQD}, we modify the theory such that the $\CP^N$ model
looks as if it can have Wess-Zumino terms with fractional coefficients and
the $\CP^1$ model looks as if it has $\theta $ that is a fractional multiple
of $\pi$.  We do that by making the  $N+1$ scalars have charge $q$ under the $U(1)$ gauge field in the original microscopic model~\eqref{Linearmodel}.

Again, with an appropriate potential the gauge symmetry is Higgsed and the
low energy theory is a nonlinear model on $\CP^N$. However, unlike the
previous case, this is not the whole story.  Now the microscopic $U(1)$ gauge
symmetry is Higgsed to $\ZZ_q$.  The coupling of this $\ZZ_q$ gauge field to
the $\CP^N$ coordinates is obtained through the equation of motion $qdb =
\omega + \cdots$.  We give a global interpretation of this coupling
in~\S\ref{subsec:2.4}.

Following \cite{MaldacenaSS,BanksZN} it is convenient to represent the unbroken $\ZZ_q$ gauge theory in terms of two $U(1)$ gauge fields $b$ and $c$ as
 \begin{equation}\label{CStp}
 \begin{aligned}
 &\int_{\CM_3}\left({q\over 2\pi} cdb + {K\over 4\pi} bdb +{1\over 2\pi} db B
- {1\over 2\pi} c\omega \right) \qquad {\rm for} \qquad K\in 2\ZZ \\
 &\int_{\CM_3}\left({q\over 2\pi} cdb +{K\over 4\pi} bdb +{1\over 2\pi} db A - {1\over 2\pi} c\omega \right) \qquad {\rm for} \qquad K\in 2\ZZ +1~.
 \end{aligned}
 \end{equation}
In the first term $c$ is a Lagrange multiplier $U(1)$ gauge field forcing $b$
to be a $\ZZ_q$ gauge field.  It can be thought of as the dual of the overall
phase of the fundamental scalars.  As explained in \cite{KapustinGUA}, the
terms proportional to $K$ are Dijkgraaf-Witten terms \cite{DijkgraafPZ}\ in
this $\ZZ_q$ gauge theory.  For $q=1$ the expressions~\eqref{CStp} reduce
to~\eqref{CSfd} since ${1\over 2\pi} cdb$ is a trivial theory where we can
use the equations of motion freely without missing global issues
\cite{WittenYA}.

The third and fourth terms in \eqref{CStp} represent couplings of the $\ZZ_q$
gauge theory to the background fields $A$ or $B$ and to the fields of the
nonlinear model through the pull back of its K\"ahler form
$\omega$.\footnote{For $\omega=0$ we indeed have a $\ZZ_q$ gauge theory, but
for nonzero $\omega$ we have $qdb=\omega$ showing that $b$ is not a $\ZZ_q$
gauge field.  We will discuss the geometric interpretation of this
construction in section 5.}  The equation of motion of $c$ sets $qdb=\omega$
in agreement with the microscopic analysis.  This determines $db$ in terms of
the nonlinear model fields, but leaves a $\ZZ_q$ gauge field undetermined.

Consider the world line $\mathcal{ C}$ of a small Skyrmion.  We can
approximate the nonlinear model configuration by $\omega = 2\pi
\delta^{(2)}(\mathcal{ C})$.  Therefore, the term with $\omega$ in
\eqref{CStp} can be replaced by a Wilson line $e^{i\int_\mathcal{ C} c}$.  A
standard computation in the TQFT \eqref{CStp}, which follows from the
equations of motion, shows that this particle carries fractional $U(1)$
charge (under $A$ or $B$), which is $1\over q$ and its spin is ${K\over
2q^2}\mod 1$.

We conclude that in this system (with $q\not=1$) the Skyrmions become anyons.
However, the total number of Skyrmions in a compact space
 \begin{equation}\label{qtotal}
  {1\over 2\pi} \int_{\CM_2} \omega = {1\over 2\pi} \int _{\CM_2} qdb \in  q
  \ZZ
 \end{equation}
must be a multiple of $q$.  In order to determine the quantum numbers of this
configuration we deform it to $\omega = 2\pi q m \delta^{(2)}(\mathcal{ C})$ with
integer $m$.  Substituting this in the TQFT \eqref{CStp} we see that all these
anyons are combined to the line $e^{iqm\int_\mathcal{ C} c}$.  This line carries
integer charge and its total spin is ${Km^2\over 2}\mod 1$; i.e.\ it is
either an integer or half-integer.

We can attempt to integrate out $b$ in \eqref{CStp}.  Such integration out is
not legal because the field $b$ has long range interactions.  If we do that
anyway, using the equation of motion $qdb = \omega + \cdots$ and ignoring the
fact that it does not determine $b$ (not even up to a gauge transformation),
the expressions \eqref{CSfd}, \eqref{CSfdb}\ are modified.

For $N>1$ \eqref{CSfd} becomes
 \begin{equation}\label{CSfdp}
 \begin{aligned}
 &\int_{\CM_4} \left({K\over 4\pi q^2} \omega\omega  +{1\over 2\pi q} \omega
dB \right)\qquad {\rm for} \qquad K\in 2\ZZ \\
&\int_{\CM_4}\left({K\over 4\pi q^2} \omega\omega +{1\over 2\pi q} \omega
dA\right) \qquad {\rm for} \qquad K\in 2\ZZ +1 ~.
 \end{aligned}
 \end{equation}
The coefficient of the Wess-Zumino term acquired a factor of $1\over q^2$ and
the coupling to the background field acquired a factor of $1\over q$.  This
means that now the Skyrmions carry charge $1\over q$, as we saw above.  For
$N=1$ \eqref{CSfdb} becomes
 \begin{equation}\label{CSfdbp}
 \begin{aligned}
 &{\pi K\over q^2} \CN +{1\over 2\pi q} \int_{\CM_3} \omega B  \qquad {\rm
for}\qquad K\in 2\ZZ\\
&{\pi K\over q^2} \CN +{1\over 2\pi q} \int_{\CM_3} \omega A \qquad {\rm
for}\qquad K\in 2\ZZ+1 ~.
 \end{aligned}
 \end{equation}
As for higher $N$, the Skyrmions have fractional charge and the low energy
$\theta$-parameter is a fractional multiple of $\pi$.

We should emphasize, however, that the expressions \eqref{CSfdp},
\eqref{CSfdbp} are useful in analyzing local properties, but they are
imprecise and do not capture the global structure correctly.  For that we
need to keep the $\ZZ_q$ gauge field and not integrate it out.

We would like to relate this construction to the discussion in
section 3.  In this system with generic $q$ we cannot place a single Skyrmion
on a spatial $S^2$ because it is not invariant under the $\ZZ_q$ gauge
symmetry.  We can, however, place $q$ Skyrmions on $S^2$.  Then their total
spin is indeed in an $SU(2)$ representations.  We can repeat this point using
Skyrmion operators.  These operators are not gauge invariant.  They carry
$\ZZ_q$ gauge charge.  As such, they might not be in $SU(2)$ representations.
We can construct a gauge invariant operator by fusing $q$ Skyrmion operators.
This object has ${1\over 2\pi}\int_{S^2} \omega = q$ and it is in an $SU(2)$
representation.

To finish our discussion of effective fractional $\theta$ terms, we would like to present another construction, where the low energy theory consists of the $\CP^1$ model coupled to a TQFT.
Let us start with the following Euclidean Lagrangian
\begin{equation}\label{CPQHE}
\sum_{I=1,2}|D_b\phi^I|^2+\mu^2 \sum_{I=1,2}|\phi^I|^2+ (\sum_{I=1,2}|\phi^I|^2)^2+ {i\over 2\pi} bdc+{iq\over 4\pi} cdc +
\begin{cases}
{i\over 2\pi}cdB \qquad {\rm
for}\qquad q\in 2\ZZ\\
{i\over 2\pi}cdA \qquad {\rm
for}\qquad q\in 2\ZZ+1~,
 \end{cases}
\end{equation}
where again $B$ is a background $U(1)$ gauge field and $A$ is a spin$^c$
connection.  Physically the gauge fields $b$ and $c$ could be interpreted as
emergent gauge fields and $c$ represents a fractional quantum Hall effect.
As above, we end up with $db=\omega+\cdots$ and the coupling to the $\CP^1$ degrees of freedom
\eqref{CSfdb} becomes
 \begin{equation}\label{CSfdbm}
 {1\over 2\pi}\omega c+{q\over 4\pi} cdc +
 \begin{cases}
{1\over 2\pi}cdB \qquad {\rm
for}\qquad q\in 2\ZZ\\
{1\over 2\pi}cdA \qquad {\rm
for}\qquad q\in 2\ZZ+1~.
\end{cases}
\end{equation}
We see here a $U(1)_q$ Chern-Simons theory of $c$ coupled to the nonlinear $\CP^1$ model and to a background field ($A$ or $B$).

As in the previous construction \eqref{CStp}, a small Skyrmion with a worldline $\mathcal{ C}$ is represented in the $U(1)_q$ theory by a line operator $e^{i\int_\mathcal{ C} c}$.  It has spin $1\over 2q$ and charge $1\over q$; i.e.\ it is an anyon.  Also as in the previous example, we can incorrectly integrate out $c$ to find an effective $\CP^1$ theory with $\theta=-\pi/q$.

More generally, instead of considering the concrete examples leading
to~\eqref{CStp}, \eqref{CSfdbm}, we could couple the gauge field $b$ to a
general Chern-Simons TQFT with some matrix of Chern-Simons couplings
$k_{ij}$.

   \section{Invertible field theories and effective actions}\label{sec:2}

  \subsection{Generalities}\label{subsec:2.1}

The (effective) action of an $n$-dimensional field theory obeys strong
locality constraints: it can be computed on a (Wick rotated) spacetime by
assembling local contributions.  For example, a typical kinetic term
$\int_{}|d\phi |^2/2$ is the integral of an expression computed locally from
a field~$\phi $.  More interesting topological terms do not have such simple
formulas yet obey many of the same locality properties.  The strongest
expression of that locality is encoded in the notion of an \emph{extended}
field theory~\cite{F1,La,Lu}.  Furthermore, the exponentiated action is the
partition function of an \emph{invertible} field theory: for example, if a
closed $n$-manifold~$M$ is cut along a codimension one submanifold~$N$, then
the vector space of ``states'' associated to~$N$ is 1-dimensional.
The notion of an invertible field theory systematizes the locality one
demands in a classical action; the partition function is the exponential of
what would be the classical action (which need only be well-defined up to
shifts by~$2\pi i$).

The structure of a field theory, in particular its locality, is captured by
an Axiom System originally introduced by Segal~\cite{Se} in the context of
2-dimensional conformal field theories and Atiyah~\cite{At} for topological
field theories.  It has since been expanded and used more generally; it is
most developed for topological theories.  In this framework an invertible
field theory, after Wick rotation, becomes a map of \emph{spectra} in the
sense of stable homotopy theory.  Recently an \emph{extended} notion of
unitarity, or rather its Wick rotated version---reflection positivity---was
introduced in the invertible case~\cite{FH}.  Of course, we expect unitarity
in any physical theory, so an invertible field theory used in the action
should be reflection positive and in this paper we restrict to such field
theories.  The theorems in~\cite{FH} classify \emph{deformation classes} of
invertible theories as well as isomorphism classes of invertible topological
theories, as we review shortly.  Two exponentiated actions are in the same
deformation class if they can be joined by a smooth path of exponentiated
actions; in physics terms, they are related by adding a local term to the
lagrangian.  For example, $t\mapsto \exp\int_{}(1-t)|d\phi |^2/2$, $0\le
t\le1$, is a path from the exponentiated kinetic action to the trivial
action.  Topologically nontrivial actions have nontrivial deformation
classes, so are detected by the stable homotopy invariants introduced below.

The main result of~\cite{FH} states that the abelian group of deformation
classes of unitary invertible $n$-dimensional field theories is isomorphic to
the abelian group of homotopy classes of maps $\sB\to\Sigma ^{n+1}\IZ$ from a
Thom spectrum\footnote{The specific Thom spectrum, which involves a choice of
tangential structure (orientation, spin structure, etc.), is determined by
the symmetries in the theory.}~$\sB$ to the shifted Anderson dual to the
sphere spectrum.  We refer to \cite[\S5]{FH} and the references therein for
exposition, and remark that the Anderson dual was introduced in this context
in~\cite{HS}.  The torsion subgroup is the group of \emph{topological}
theories.\footnote{In~\cite{FH} this statement about topological theories is
a theorem; the general statement is left as a conjecture because in that
paper the authors did not set up the mathematical infrastructure necessary to
make it a theorem.}  We will also consider unitary invertible $n$-dimensional
theories with partition function an integer; they are classified up to
isomorphism by a homotopy class of maps $\sB\to\Sigma ^n\IZ$.  (For a theory
of oriented or spin manifolds, `unitary' here means that the partition
function changes sign under orientation-reversal.)  This classification
statement is not proved in~\cite{FH}, nor do we give a full discussion here,
but in any case it only enters peripherally in what follows.

The abelian groups computed here are generalized cohomology groups for the
cohomology theory defined by the Anderson dual.  They are not \emph{homotopy}
groups of a \emph{space}, but rather generalized \emph{cohomology}
groups of a \emph{spectrum}.

In the following subsections we treat the topological terms for the nonlinear
$\CP^N$~model on spin and $\spinc$ manifolds.  Then
in~\S\ref{subsec:2.4} we comment briefly on the effective models
in~\S\ref{sec:6}.

  \subsection{Spin manifolds}\label{subsec:2.2}

We begin by defining the Wess-Zumino term which appears in~\eqref{CSfd}.  We
express it in terms similar to the WZ term in the effective action for
pions~\cite{F2} and the spin Chern-Simons action~\cite{BM,J}.\footnote{It is
exactly the spin $U(1)$~Chern-Simons action at the lowest level, but
restricted to $\CP^N\subset \CP^{\infty}\simeq BU(1)$;
see~Remark~\ref{thm:14}.}  Namely, we use a generalized cohomology theory~$E$
which is a 2-stage Postnikov truncation of~$\IZ$: there is a map $E\to\IZ$
which captures the top two nonzero homotopy groups of the co-connective
spectrum~$\IZ$.  We defer to~\cite[\S1]{F2} for details about~$E$, which we
freely use in the following.  Two salient points are (i)~$E$~is
spin-oriented, so $E$-cohomology classes can be evaluated on spin manifolds,
and (ii)~there is a long exact sequence
  \begin{equation}\label{eq:1}
     \cdots\longrightarrow H^q(X;\ZZ)\xrightarrow{\;\;i\;\;}
     E^q(X)\xrightarrow{\;\;j\;\;}
     H^{q-2}(X;\zt) \xrightarrow{\;\;\beta \circ Sq^2\;\;}
     H^{q+1}(X;\ZZ)\longrightarrow \cdots
  \end{equation}
for any space~$X$, where $\beta $~is the integer Bockstein map.  This long
exact sequence characterizes~$E$.

  \begin{lemma}[]\label{thm:1}
 For~$N\ge2$ there is an isomorphism $E^4(\CP^N)\cong \ZZ$; the homomorphism
$H^4(\CP^N;\ZZ)\longrightarrow E^4(\CP^N)$ maps a generator to twice a
generator.  Also, $E^4(\CP^1)\cong \zt$ and a generator of~$E^4(\CP^N)$
restricts to a generator of~$E^4(\CP^1)$ under a linear inclusion
$\CP^1\hookrightarrow \CP^N$.
  \end{lemma}

  \begin{proof}
 The first statement is part of the proof of \cite[Proposition~1.9]{F2},
where it is also shown that the generator is the characteristic
class~$\lambda \in E^4(BSO)$ of the real 2-plane bundle underlying
$\mathcal{O}(1)\to\CP^N$.  Restricting that bundle under
$\CP^1\hookrightarrow \CP^N$ we obtain the second statement, after applying
the long exact sequence~\eqref{eq:1} with~$q=4$ and $X=\CP^1$.
  \end{proof}

\noindent
 Fix~$N\in \ZZ^{\ge1}$ and let $\chi $~denote a generator of~$E^4(\CP^N)$.
It has a unique lift $\cc\in \cE^4(\CP^N)$ to the differential theory, as we
see from~\cite[(1.8)]{F2}.  Let $M$~be a closed spin 3-manifold equipped with
a smooth map $\phi \:M\to\CP^N$.  The following is an exact analog
of~\cite[Definition~4.1]{F2}.

  \begin{definition}[]\label{thm:2}
 The \emph{WZ factor} in the $\sigma $-model exponentiated action on spin
manifolds is
  \begin{equation}\label{eq:2}
     W\!\mstrut _M(\phi ) = \exp\left( 2\pi i\;\pi ^M_*\phi ^*\cc \right).
  \end{equation}
  \end{definition}

\noindent
 The projection $\pi ^M\:M\to\pt$ induces the pushforward $\pi
^M_*\:\cE^4(M)\to \cE^1(\pt)\cong \RZ$.  We emphasize that
Definition~\ref{thm:2} works for~$N=1$ as well as~$N\ge2$.

  \begin{remark}[]\label{thm:7}
 Formula~\eqref{eq:2} corresponds to~$K=1$ in~\S\ref{sec:4}; the formula for
arbitrary~$K$ multiplies the quantity in parentheses by~$K$.  An application
of Stokes' theorem for differential $E$-theory analogous to~\cite[(4.3)]{F2}
reproduces the WZ term in~\eqref{CSfd}---the first term in the formula
with~$K=1$---assuming that $M$~bounds a compact spin 4-manifold~$W$ and $\phi
$~extends to a map $\Phi \:W\to\CP^N$.  (See Remark~\ref{thm:6} below.)
  \end{remark}

  \begin{remark}[]\label{thm:8}
 As in~\cite[(4.10)]{F2} the $\sigma $-model with WZ factor encodes the
statistics of skyrmions.
  \end{remark}

  \begin{remark}[]\label{thm:14}
 The hyperplane bundle $\Oo\to\CP^N$ has a natural $SU(2)$-invariant metric
and connection.  The WZ~factor~\eqref{eq:2} is the lowest level spin
Chern-Simons invariant of its pullback via~$\phi $.
  \end{remark}

The WZ factor varies smoothly with~$\phi $ if~$N\ge 2$, but is a topological
invariant if~$N=1$.  Next, we give a topological description for~$N=1$ which
does not use $E$-cohomology.  Let $M$~be a closed spin 3-manifold and $\phi
\:M\to\Co$.  Fix a regular value~$p\in \Co$ and a basis $e_1,e_2$ of
$T_p\Co$.  This produces a normal framing of the 1-manifold $S:=\phi \inv
(p)$, which after applying Gram-Schmidt we can assume is orthonormal.  (The
contractible choice of a Riemannian metric on~$M$ does not affect the mod~2
invariant we are defining.)  At each point of~$S$ there is a unique
completion~$e_1,e_2,e_3$ to an oriented orthonormal basis, and so two lifts
to the $\Spin_3$-bundle of frames of~$M$.  The resulting double cover of~$S$
may be identified with the spin bundle of frames of a spin structure on~$S$.

  \begin{lemma}[]\label{thm:3}
Set~$N=1$.  Then

 \begin{enumerate}[label=\textnormal{(\roman*)}]

 \item $W\!\mstrut _M(\phi ) =(-1)^{[S]}$, where $[S]\in \Omega
_1^{\Spin}\cong \zt$ is the spin bordism class of~$S$.

 \item  $W\!\mstrut _{S^3}(\phi )$ is the mod~2 Hopf invariant of~$\phi
\:S^3\to\Co$.

 \end{enumerate}
  \end{lemma}

\noindent
 $S$~is a finite union of spin circles, each of which is bounding
(Neveu-Schwarz) or nonbounding (Ramond).  The invariant in~(i) is~$\pm1$
depending on the parity of the number of nonbounding components.
Alternatively, the normal framing of~$S$ determines an element of framed
bordism, which is isomorphic to spin bordism in dimension one.
Assertion~(ii) shows that $W\!\mstrut _M(\phi )$~extends the $\theta =\pi $
term for the Hopf invariant.

  \begin{proof}
 The inclusion $\iota\:\{p\}\hookrightarrow \Co$ with normal
framing~$e_1,e_2$ induces a pushforward $\iota _*\:E^2(\{p\})\longrightarrow
E^4(\Co)$.  Let $\alpha \in E^2(\{p\})\cong \zt$ be the generator and $\ca\in
\cE^2(\{p\})$ the unique lift to the differential group.  The commutativity
of the diagram
  \begin{equation}\label{eq:3}
     \begin{gathered} \xymatrix{E^2(\{p\})\ar[r]^{\iota _*} \ar[d]_{j} &
     E^4(\Co)\ar[d]^{j} \\ H^0(\{p\};\zt)\ar[r]^{\iota _*} & H^2(\Co;\zt)}
     \end{gathered}
  \end{equation}
implies that $\iota _*\alpha =\chi $.  It follows that $\pi ^{M}_*\phi ^*\cc
= (\pi ^{S})_*(\pi ^S) ^*(\ca )$.  The right hand side contains $\pi
^S_*\:\cE^2(S)\to \cE^1(\pt)$, which equals the topological pushforward
  \begin{equation}\label{eq:4}
     \pi ^S_*\:E^1(S;\RZ)\to E^0(\pt;\RZ)
  \end{equation}
on the flat part\footnote{which is a spectrum~$E_{\RZ}$ with $\pi
_0E_{\RZ}\cong \RZ$ and $\pi _{-1}E_{\RZ}\cong \zt$ connected by a nontrivial
$k$-invariant} of the differential $E$-theory groups, after identifying
$\ca\in E^1(\pt;\RZ)\cong \zt$.  Rewrite~\eqref{eq:4} as $\pi^S
_*\:ko^{-3}(S;\RZ)\to ko^{-4}(\pt;\RZ)$, identify~$\ca$ with the generator
$a\in ko^{-3}(\pt;\RZ)\cong \zt$, and write $\pi ^S_*\bigl((\pi ^S)
^*(a)\bigr)=\pi ^S_*(1)\,a$.  Finally, (i)~follows from $\pi ^S_*(1)=[S]\in
ko^{-1}(\pt)$ and \cite[(B.10)]{FMS}.

For~(ii) we observe that both $\phi \mapsto\pi^{M}_*\phi ^*\cc$ and the mod~2
Hopf invariant define homomorphisms $\pi _3(\Co)\to\zt$, so it suffices to
verify the equality on the Hopf map $S^3\to\Co$, which amounts to
verifying that the induced spin structure on a fiber of the Hopf map is
nonbounding, or equivalently the normal framing is nontrivial.  That follows
since the Hopf map represents the nontrivial element of~$\Omega
_1^{\textnormal{framed}}$ via the Pontrjagin-Thom construction.
  \end{proof}

We now turn to the classification of possible topological terms, so to
bordism computations of invertible field theories of spin 3-manifolds
equipped with a map to~$\Co$.  If $\sB$~is any spectrum then there is a short
exact sequence~\cite[(B.3)]{HS}
  \begin{equation}\label{eq:5}
     0\longrightarrow \Ext^1(\pi _{q-1}\sB,\ZZ)\longrightarrow [\sB,\Sigma
     ^q\IZ]\longrightarrow \Hom(\pi _q\sB,\ZZ) \longrightarrow 0
  \end{equation}
where $[\sX,\sY]$ denotes the group of homotopy classes of spectrum maps
$\sX\to\sY$.  The Thom bordism spectrum of spin manifolds equipped with a map
to~$\Co$ is\footnote{The `$+$'~denotes a disjoint basepoint, which occurs
since the Thom spectrum of $W\oplus 0\to B\!\Spin\times \Co$ is the smash
product of the Thom spectra of the universal bundle $W\to B\!\Spin$ and
$0\to\Co$; the latter is~$\Co_+$.}
  \begin{equation}\label{eq:6}
     \MSpin\wedge \Co_+ \;\simeq\; \MSpin \wedge \Co\,\vee\,\MSpin\;\simeq\;
     \Sigma ^2\MSpin \,\vee\,\MSpin.
  \end{equation}
We first ask if there is an integer-valued invertible field
theory\footnote{If $\nu $~in \eqref{anut} is an integer invariant of spin
manifolds which extends the Hopf invariant, then to use it in the action of a
quantum field theory it should be fully local, hence the partition function
of an invertible integer-valued topological field theory.} whose partition
function extends the Hopf invariant, so a map $\MSpin\wedge \Co_+\to\Sigma
^3\IZ$.  (We use the unproved assertion towards the end
of~\S\ref{subsec:2.1}.)  It follows from~\eqref{eq:6} and low dimension spin
bordism that $[\MSpin\wedge \Co_+,\Sigma ^3\IZ]\cong \zt$, and so the
partition function of any integer-valued theory vanishes.  In particular,
there is no possibility to write~$e^{i\nu \theta }$ in an exponentiated
action if $\nu$~specializes on~$S^3$ to the Hopf invariant.

More directly, we can compute the group of deformation classes of
\emph{topological} field theories.

  \begin{theorem}[]\label{thm:4}
 The group of deformation classes of unitary invertible topological field
theories of spin 3-manifolds equipped with a map to~$\Co$ is isomorphic
to~$\zt$; it is also the group of isomorphism classes of theories of this
type.  The generator has partition function~\eqref{eq:2}.
  \end{theorem}

\noindent
 There are non-topological invertible field theories whose partition function
is an exponentiated $\eta $-invariant.  They depend on a metric (but not on
the map~$\phi $ to~$\Co$).  Theorem~\ref{thm:4} rules out values of~$\theta $
other than $\theta =0$ and~$\theta =\pi $, as deduced in~\S\ref{sec:5} by a
different argument.

  \begin{proof}
 According to the main theorem in~\cite{FH} the group in the theorem is
the torsion subgroup of $[\MSpin\wedge \Co_+,\Sigma ^4\IZ]$, which is the
$\Ext^1$~group in~\eqref{eq:5}.  Using~\eqref{eq:6} we deduce
  \begin{equation}\label{eq:7}
     \pi _3\bigl(\MSpin\wedge \Co_+\bigr)\cong \pi _1\MSpin\cong \zt
  \end{equation}
since $\Ext^1(\zt,\ZZ)\cong \zt$.  The isomorphism $\pi _3\bigl(\MSpin\wedge
\CP^1_+ \bigr)\xrightarrow{\;\cong \;} \pi _1\MSpin$ maps a spin
3-manifold~$M$ equipped with a map $\phi \:M\to\Co$ to the inverse image of a
regular value, so the identification of the partition function follows from
the construction before Lemma~\ref{thm:3}.
  \end{proof}

  \begin{remark}[]\label{thm:6}
 The nontriviality of the bordism group~\eqref{eq:7} shows that the WZ
factor~\eqref{eq:2} cannot always be computed by extending over a bounding
4-manifold, for example for $M=S^3$ equipped with the Hopf map $\phi
\:S^3\to\Co$, as was done in~\eqref{CSfd}.
  \end{remark}

  \subsection{$\operatorname{\bf Spin}^c$ manifolds}\label{subsec:2.3}

We repeat the analysis for $\spinc$ manifolds.  The cohomology theory~$E$
does not have a $\spinc$ orientation---$E$-cohomology classes cannot be
integrated over $\spinc$ manifolds---so differential $E$-cohomology does not
enter our analysis.  On the other hand, complex $K$-theory is
$\spinc$-oriented, and so we use differential $K$-theory, but only implicitly
as we express the integral of a differential $K$-theory class as an
exponentiated $\eta $-invariant~\cite{K,O}.

Recall that the group $\Spcn$ is a central extension $1\to\TT\to\Spcn\to
SO_n\to 1$ of the special orthogonal group $SO_n$ by the circle group $\TT$
of phases.  A $\spinc$ manifold~$M$ is an oriented Riemannian manifold
equipped with a principal $\Spcn$-bundle $P\to M$ lifting its oriented
orthonormal bundle of frames and a connection---the \emph{spin$^c$
connection}---on $P\to M$ compatible with the Levi-Civita connection.  There
is a homomorphism $\Spcn\to\TT$ and so an associated circle bundle with
connection over~$M$, called the \emph{characteristic bundle}.  A $\spinc$
manifold has a canonical Dirac operator.  A spin structure on a $\spinc$
manifold is equivalent to a flat trivialization of its characteristic bundle.

The analog of Definition~\ref{thm:2} on a $\spinc$ manifold depends on the
$\spinc$ connection and the map~$\phi \:M\to\CP^N$, but not on the Riemannian
metric.  It uses the $\eta $-invariant of Atiyah-Patodi-Singer~\cite{APS}.
Recall that these authors define a more refined invariant~$\xi =(\eta
+\dim\ker)/2$ and that the \emph{exponentiated $\eta $-invariant }~$\exp(2\pi
i\xi )$ varies smoothly with parameters.  Let $M$~be a closed $\spinc$
3-manifold equipped with a smooth map $\phi \:M\to\CP^N$.  Let
$\mathcal{O}(1)\to \CP^N$ be the hyperplane line bundle with its standard
covariant derivative.

  \begin{definition}[]\label{thm:9}
  The \emph{WZ factor} in the $\sigma $-model exponentiated action on
$\spinc$ manifolds is the exponentiated $\eta $-invariant of the Dirac
operator coupled to the virtual bundle $\phi ^*\Oo-1$.
  \end{definition}

\noindent
 This is the ratio of the exponentiated $\eta $-invariant of the $\spinc$
Dirac operator coupled to $\phi ^*\Oo$ and the exponentiated $\eta
$-invariant of the uncoupled $\spinc$ Dirac operator.

We claim that this reproduces~\eqref{CSfd} for~$K=1$ in case $M$~bounds a
$\spinc$ 4-manifold~$W$ equipped with a map $\Phi \:W\to\CP^N$ which
extends~$\phi $.  In that case the main theorem in~\cite{APS} computes the WZ
factor as the exponential of the integral over~$W$ of the Chern-Weil
differential forms which represent the degree~4 term in
  \begin{equation}\label{eq:8}
     \Ahat(W)e^{c/2}\bigl(e^x-1 \bigr).
  \end{equation}
In this expression $c,x\in H^2(W;\ZZ)$ are the Chern classes of the
characteristic line bundle and $\Phi ^*\Oo$, respectively; each has degree~2.
The degree~4 term is $(x^2+xc)/2$, which matches the differential form
expression~\eqref{CSfd}.

  \begin{remark}[]\label{thm:10}
 Even for~$N=1$ the WZ factor on a $\spinc$ manifold depends on the $\spinc$
connection and the map~$\phi $ (not just up to homotopy), so is not a
topological invariant.
  \end{remark}

  \begin{remark}[]\label{thm:12}
 The exponentiated $\eta $-invariant is the partition function of an extended
invertible unitary field theory, so a valid factor in an exponentiated
action.  As a nonextended field theory of 2-~and 3-manifolds this follows
from the theorems in~\cite{DF}.  To construct the extended theories we
can use differential $K$-theory, following the ideas in~\cite{HS}.
  \end{remark}

  \begin{proposition}[]\label{thm:11}
 The $\spinc$ WZ factor in Definition~\ref{thm:9} reduces to the spin WZ
factor~\eqref{eq:2} on a spin manifold.
  \end{proposition}

\noindent
 In particular, by Lemma~\ref{thm:3}(ii), it extends the mod~2 Hopf
invariant.

  \begin{proof}
 Suppose first that the spin 3-manifold~$M$ bounds a spin 4-manifold~$W$ over
which $\phi $~extends.  A spin structure on a $\spinc$ manifold trivializes
the characteristic class~$c$ of the characteristic bundle, so the degree~4
term in~\eqref{eq:8} reduces to~$x^2/2$.  The corresponding statement about
differential forms follows since the curvature of the characteristic bundle
vanishes.  Then the integral over~$W$ which computes the $\spinc$ WZ factor
reduces to the one for the spin WZ factor alluded to in Remark~\ref{thm:7}.
This proves the proposition in the bounding case.

If $N=1$~then \eqref{eq:7}~computes the relevant bordism group to be cyclic
of order~2 with generator the Hopf map $\phi \:S^3\to\Co$ (see
Remark~\ref{thm:6}), so we cannot directly apply the argument in the previous
paragraph.  Observe, however, that if $\CP^1\hookrightarrow \CP^N$ is a
linear embedding then the bundle $\Oo\to\CP^N$ with its covariant derivative
restricts to the bundle $\Oo\to\Co$ with its covariant
derivative.\footnote{The hermitian metric and holomorphic structure restrict,
so too does the resultant Chern covariant derivative.}  Since $\cc$~in
Definition~\ref{thm:2} is a differential characteristic class
of~$\Oo\to\CP^N$, it follows that both WZ factors are unchanged by composing
$\phi \:M\to\Co$ with the embedding $\CP^1\hookrightarrow \CP^N$.  So it
suffices to prove that any spin 3-manifold~$M$ equipped with a map $\phi
\:M\to\CP^N$ bounds for any~$N\ge2$.  The map~$\phi $ can be homotoped into
the 4-skeleton~$\CP^2$, so it suffices to take~$N=2$.

First, arguing as in~\eqref{eq:6} we are reduced to showing $A:=\pi
_3\bigl(\MSpin\wedge \CP^2 \bigr)=0$.  The cofibration sequence
$S^3\xrightarrow{\eta } S^2\to \CP^2$ gives rise to the exact sequence
  \begin{equation}\label{eq:10}
      \pi _3\bigl(\MSpin\wedge S^3
     \bigr)\xrightarrow{\;\;\eta \;\;} \pi _3\bigl( \MSpin\wedge
     S^2\bigr)\longrightarrow \pi _3\bigl(\MSpin\wedge \CP^2
     \bigr)\longrightarrow \pi _2(\MSpin\wedge S^3),
  \end{equation}
which simplifies to $\pi _0\MSpin\to\pi _1\MSpin\to A \to 0$.  The Hopf
map~$\eta $ induces a surjective map $\pi _0\MSpin\to\pi _1\MSpin$ (as
stated above it represents the generator of stable framed bordism), from
which we conclude~$A=0$.
  \end{proof}

Next, we observe that there is no unitary integer-valued invertible field
theory of $\spinc$ manifolds equipped with a map to~$\Co$ whose partition
function specializes to the Hopf invariant.  For if there were, it would
restrict to a theory of spin manifolds with that property and that was ruled
out after~\eqref{eq:6}.  Similarly, it follows from Theorem~\ref{thm:4}
that there is no invertible $\CC^\times $-valued $\spinc$ theory which
specializes to a theory whose partition function for~$\phi \:S^3\to\Co$
is~$e^{i\nu \theta }$, where $\nu $~is the Hopf invariant and $\theta \not=
0,\pi $; any such would restrict to a spin theory with those properties.

Finally, we justify that~\eqref{CSfd} can be used in the $\spinc$ case by
computing that every 3-dimensional $\spinc$ manifold~$M$ equipped with a map
$\phi \:M\to\CP^N$ bounds, i.e., $B^{(N)}:=\pi _3(\MSpin^c\wedge \CP^N_+)=0$.
First, as in previous arguments we reduce the case~$N>2$ to~$N=2$, and we can
omit the disjoint basepoint~`$+$' since $\pi _3\MSpin^c=0$.  For~$N=2$ we use
the exact sequence~\eqref{eq:10} with $\MSpin^c$ replacing~$\MSpin$, and
since $\pi _1\MSpin^c=0$ we deduce~$B^{(N)}=0$ for~$N\ge2$.  For~$N=1$ we
also see $B^{(1)}\cong \pi _1\MSpin^c=0$.

  \subsection{The variation in~\S4}\label{subsec:2.4}

The long distance theory derived in~\S4 has fields (i)~a map $\phi
\:M\to\CP^N$, (ii)~a line bundle with connection $L\to M$, and
(iii)~ an isomorphism
  \begin{equation}\label{eq:11}
      \theta \:\phi ^*\Oo\xrightarrow{\;\cong \;}L^{\otimes q}
  \end{equation}
of line bundles with connection.  The set of isomorphism classes of
pairs~$(L,\theta )$ is a torsor over the set of isomorphism classes of
principal $\zmod q$-bundles: more precisely,\footnote{The following analog
may be helpful: a spin structure on a Riemann surface~$\Sigma $ is a
pair~$(L,\theta )$ where $L\to\Sigma $ is a holomorphic line bundle and
$\theta \:K_\Sigma \xrightarrow{\;\cong \;}L^{\otimes 2}$ an isomorphism of
the canonical bundle of~$\Sigma $ with the square of~$L$.  Any two spin
structures are related by a double cover.}  given two pairs $(L,\theta
),\,(L',\theta ')$ there is a canonical $\zmod q$-bundle $Q\to M$ such that
$(L',\theta ')\cong (L,\theta )\otimes Q$.  The $\spinc$ WZ factor has an
easy generalization in this case.

  \begin{definition}[]\label{thm:13}
 The \emph{$q $-WZ factor} in the $\sigma $-model exponentiated action on
$\spinc$ manifolds is the exponentiated $\eta $-invariant of the Dirac
operator coupled to the virtual bundle $L-1$.
  \end{definition}

\noindent
 Assuming the 3-manifold~$M$ and all its geometric data extend over a compact
4-manifold with boundary~$M$, then \eqref{eq:8}~is modified by substituting
$x\to x/q$, the consequence of~\eqref{eq:11} on Chern classes.  The degree~4
term is then $(x^2/q^2 + xc/q)/2$, which matches \eqref{CSfdbp}.

As in Proposition~\ref{thm:11} we can restrict this definition to spin
manifolds, or alternatively take the direct image of the differential
$E$-characteristic class~ $\lambda $ of the line bundle with connection~$L\to
M$, analogous to Definition~\ref{thm:2}.  (See the proof of
Lemma~\ref{thm:1}.)

  \subsection*{Acknowledgments}\label{subsec:2.5}

We would like to thank Luis Alvarez-Gaume and Paul Wiegmann for useful
discussions.  This work started during a workshop at the Simons Center for
Geometry and Physics.  We thank the SCGP for providing the nice and
stimulating environment that facilitated this collaboration.  D.F. is
supported in part by the National Science Foundation under Grant Number
DMS-1611957.  He completed part of this paper at the Aspen Center for
Physics, which is supported by National Science Foundation grant PHY-1607611.
Z.K. is supported in part by an Israel Science Foundation center for
excellence grant and by the I-CORE program of the Planning and Budgeting
Committee and the Israel Science Foundation (grant number 1937/12). Z.K. is
also supported by the ERC STG grant 335182 and by the Simons Foundation grant
488657 (Simons Collaboration on the Non-Perturbative Bootstrap). N.S. is
supported in part by DOE grant DE-SC0009988.  Any opinions, findings, and
conclusions or recommendations expressed in this material are those of the
authors and do not necessarily reflect the views of the National Science
Foundation or the DOE.


\providecommand{\href}[2]{#2}\begingroup\raggedright\endgroup

  \end{document}